\documentclass[conference]{IEEEtran}
\IEEEoverridecommandlockouts
\usepackage{cite}
\usepackage{amsmath,amssymb,amsfonts}
\usepackage{graphicx}

\newcommand{\bv}{\mathbf{v}}    % transmit signal vector
\newcommand{\bx}{\mathbf{x}}    % transmit signal vector
\newcommand{\bH}{\mathbf{H}}    % channel matrix
\newcommand{\bn}{\mathbf{n}}    % noise vector
\newcommand{\bN}{\mathbf{N}}    % noise covariance
\newcommand{\bC}{\mathbf{C}}    % coupling matrix
\newcommand{\bI}{\mathbf{I}}    % Current vector or identity matrix
\newcommand{\bV}{\mathbf{V}}    % voltage vector
\newcommand{\bA}{\mathbf{A}}    % auxiliary matrix
\newcommand{\define}{\triangleq}    % Bold vector b

\newcommand{\mbbC}{\mathbb{C}}	% field of complex number
\newcommand{\mccn}{\mathcal{CN}}	% 	complex Gaussian

\newcommand{\btt}{\bsym{\theta}}

\newcommand{\bzro}{\mathbf{0}}

\newcommand{\ben}{\begin{enumerate}} 	  	% 	Begin Enumerate
\newcommand{\een}{\end{enumerate}} 			% 	End Enumerate
\newcommand{\beq}{\begin{equation}} 	  	% 	Begin Equation
\newcommand{\eeq}{\end{equation}} 			% 	End Equation
\newcommand{\bes}{\begin{equation*}}
\newcommand{\ees}{\end{equation*}}

\newcommand{\bea}{\begin{eqnarray}}		% 	Begin Equation Array
\newcommand{\eea}{\end{eqnarray}} 		% 	End Equation Array
\newcommand{\beas}{\begin{eqnarray*}}
\newcommand{\eeas}{\end{eqnarray*}}
\newcommand{\ba}{\begin{array}}
\newcommand{\ea}{\end{array}}

\newcommand{\sbea}{\nopagebreak[3]\samepage\begin{eqnarray}}
\newcommand{\seea}{\end{eqnarray}\pagebreak[0]}
\newcommand{\sbeas}{\nopagebreak[3]\samepage\begin{eqnarray*}}
\newcommand{\seeas}{\end{eqnarray*}\pagebreak[0]}

\newcommand{\er}[1]{{\rm(\ref{#1})}}

\newcommand{\bit}{\begin{itemize}}
\newcommand{\eit}{\end{itemize}}
\newcommand{\bsym}{\boldsymbol}

\newcommand{\ti}{\textit}
\newcommand{\nn}{\nonumber}

\DeclareMathOperator{\Trace}{Tr}

\newtheorem{theorem}{Theorem}%[section]

\newenvironment{proof}[1][Proof]{\begin{trivlist}
\item[\hskip \labelsep {\bfseries #1}]}{\end{trivlist}}

\newcommand{\qed}{\nobreak \ifvmode \relax \else
      \ifdim\lastskip<1.5em \hskip-\lastskip
      \hskip1.5em plus0em minus0.5em \fi \nobreak
      \vrule height0.75em width0.5em depth0.25em\fi}

\def\iflatex{\iftrue}
\def\ifcomments{\iffalse}

\begin{document}

\title{A Hybrid Approach to Joint Estimation of Channel and Antenna impedance
	\thanks{This material is based
		upon work supported by the National Science Foundation under Grant
		1343309. Any opinions, findings, and conclusions or
		recommendations expressed in this material are those of the
		author(s) and do not necessarily reflect the views of the National
		Science Foundation.}
}
\author{\IEEEauthorblockN{Shaohan~Wu}
	\IEEEauthorblockA{\textit{Department of Electrical and Computer Engineering} \\
		\textit{North Carolina State University}\\
		Raleigh, NC 27695-7911 \\
		{swu10}@ncsu.edu}
	\and
	\IEEEauthorblockN{Brian~L.~Hughes}
	\IEEEauthorblockA{\textit{Department of Electrical and Computer Engineering} \\
		\textit{North Carolina State University}\\
		Raleigh, NC 27695-7911 \\
		{blhughes}@ncsu.edu}
}
\maketitle
\begin{abstract}
%Over the past two decades, several works have demonstrated that impedance matching between the receive antenna and front-end can significantly impact capacity in wireless channels. In order to implement capacity-optimal matching, the receiver must know the antenna impedance. However, this impedance often varies with near-field loading in an unpredictable way. Recently, several authors have proposed adaptive matching techniques that estimate and compensate for variations in impedance. Most wireless systems have extensive resources devoted to channel estimation, such as training sequences and pilot symbols, but no mechanism to estimate impedance. %The aim of this paper is to explore the impact of diverting some of these resources to estimate antenna impedance.

This paper considers a hybrid approach to joint estimation of channel information and antenna impedance, for single-input, single-output channels. 
Based on observation of training sequences via synchronously switched load at the receiver, we derive joint maximum a posteriori  and maximum-likelihood (MAP/ML) estimators for channel and impedance over multiple packets. We investigate important properties of these estimators, e.g., bias and efficiency. We also explore the performance of these estimators through numerical examples. %, as well as the impact of channel correlation on estimation accuracy. %Our numerical results suggest that antenna impedance can be accurately estimated, in exchange for a small, controlled increase in channel estimation error.
\end{abstract}


\section{Introduction}
Antenna impedance matching at mobile receivers has been shown to significantly impact capacity and diversity of wireless channels \cite{carlo10,carlo12,lau,wall,gans,gans2}. 
This matching becomes challenging as antenna impedance changes with time-varying near-field loading, e.g., human users \cite{Lau, Ali, Thomp}. 
To mitigate this variation, joint channel and antenna impedance estimators have been derived under classical estimation assumptions \cite{Wu,Witt}. 
%Adaptive matching \cite{Lau,Thomp,Ali} and impedance estimation \cite{Ali,Wu,Witt} have been proposed to mitigate the impact of time-varying antenna impedance due to near-field loading. In a prequel, joint ML estimators for channel information and antenna impedance have been derived \cite{Wu}.
However, important properties of these estimators, e.g., bias and efficiency, remain unclear.

In this paper, we develop a hybrid estimation framework of channel information and antenna impedance, for single-input, single-output channels. %The aim is to utilize  extensive resources available in most wireless systems for channel estimation to estimate antenna impedance as well. We assume the transmitter sends a known training sequence, during which the receiver varies its own impedance in a known way. 
In particular, channel information is modeled as complex Gaussian while antenna impedance is deterministic. 
Based on observation of training sequences via synchronously switched load at the receiver, we derive the joint maximum a  posteriori  and maximum-likelihood (MAP/ML) estimators for channel and impedance over multiple packets. Then, bias, consistency, and efficiency of these estimators are investigated. We also explore the performance of these estimators through numerical examples. %, as well as the impact of channel correlation on estimation accuracy. %Our numerical results suggest that antenna impedance can be accurately estimated, in exchange for a small, controlled increase in channel estimation error.

The rest of the paper is organized as follows. We present the system model in Sec.~II, and derive the joint MAP/ML estimators for the channel and antenna impedance in Sec.~III. %Then,  the these estimators are extended over multiple packets in Sec.~IV. 
Important properties, e.g., bias and efficiency, are  studied in Sec.~IV. We then explore the performance of these estimators through numerical examples in Sec.~V, and  conclude  in Sec.~VI.

\section{System Model}
Consider a narrowband communications link with one transmit antenna and one receive antenna. We adopt the receiver model illustrated in Fig.~\ref{fig_SystemModel}, which has been widely used to model scenarios, where amplifier noise dominates at the receiver \cite{gans2,lau,wall}. %More complex models that consider multiple noise sources and impedance matching are considered in \cite{carlo10,carlo12}.
\iflatex
\begin{figure}[t!]
	\begin{center}
		\includegraphics[width=0.45\textwidth, keepaspectratio=true]{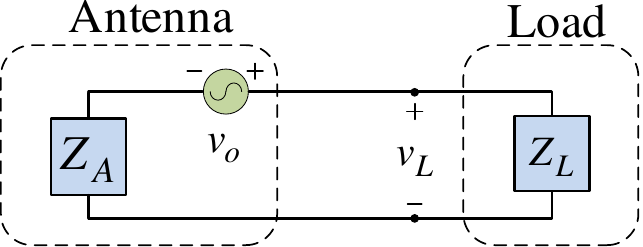}
	\end{center}
	\vspace{-5pt}
	\caption{Circuit model of a single-antenna receiver}
	\vspace{-10pt}
	\label{fig_SystemModel}
\end{figure}
\fi

In Fig.~\ref{fig_SystemModel}, the antenna is modeled by its Thevenin equivalent
\beq
v \ = \ Z_A i + v_o \ , \label{antenna}
\eeq
where $v, i \in \mathbb{C}$ are the voltage across, and current into, the antenna terminals. The antenna impedance is 
\beq\label{eq_impedance}
Z_A \ = \ R_A + j X_A \ ,
\eeq
where $R_A$ and $X_A$ are the resistance and reactance, respectively. In \er{antenna}, $v_o \in \mathbb{C}$ is the open-circuit voltage induced by the incident signal field, which can be modeled in a flat-fading environment as \cite{carlo12}
\beq\label{eq_oc_volt}
v_o \ = \ Gx \ ,
\eeq
where $x\in\mbbC$ is the transmitted symbol and $G\in\mbbC$ is the fading  path gain.
Suppose the fading path gain $G$ and antenna impedance $Z_A$ are unknown to the receiver.
We assume the transmitter sends a predetermined training sequence, $x_1 , \ldots , x_T$, and the receiver shifts synchronously through a sequence of known
impedances $Z_{L,1} , \ldots , Z_{L,T}$. If $G$ and $Z_A$ are modeled as fixed over the duration of the training sequence, the observations are given by
\bea\label{observations}
{v}_{L,t} &=& \frac{Z_{L,t} Gx_t}{Z_A + Z_{L,t}} + n_t \ , ~~~\ t = 1, \ldots , T \ ,
\eea
where ${ n}_t \sim\mathcal{CN}({0}, \sigma_n^2)$ are independent and identically distributed (i.i.d.) complex Gaussian noise. In this paper, we consider load impedances that take on only two possible values \cite{Wu},
\begin{eqnarray}\label{impedanceshifts}
	Z_{L,t} \ = \ \left\{ \begin{array}{cc}
		Z_1 , & 1 \leq t \leq K \ , \\
		Z_2 , & K < t \leq T \ .
	\end{array} \right. 
\end{eqnarray}
where $Z_1$ and $Z_2$ are known. Here we assume $Z_L=Z_1$ is the load impedance used to receive the transmitted data, and is matched to our best estimate of $Z_A$; additionally $Z_L=Z_2$ is an impedance variation introduced in order to make $Z_A$ observable.

The effective channel information needed by the communication algorithms is normally defined as \cite[eq.~13]{Wu},
\beq\label{h}
H \ \define \ \frac{Z_1G}{Z_A + Z_1} \ .
\eeq
With this definition and \er{impedanceshifts}, the sequence of observations in \eqref{observations} can be described compactly in vector form,
\beq\label{vL}
\bv_L \ \define  \ \begin{bmatrix}
	\bv_1\\
	\bv_2
\end{bmatrix} \ = \ \begin{bmatrix}
H\bx_1\\FH\bx_2
\end{bmatrix} + \bn \ ,
\eeq
where the noise is i.i.d. complex Gaussian, $\bn \sim\mathcal{CN}({0}, \sigma_n^2\bI_T)$, the two partitions of each training sequence are,
\beq
\bx_1 \ \define \  [x_1, \dots,x_K]^T\ , ~~\bx_2 \ \define \ [x_{K+1}, \dots,x_T]^T \ ,
\eeq
and $F$ is a one-to-one mapping of $Z_A$ provided $Z_1 \neq Z_2$, 
\beq\label{F}
F \ \define \ \frac{1+Z_A/Z_1}{1+Z_A/Z_2} \ .
\eeq
Thus, knowing $F$ is equivalent to knowing $Z_A$. 

%Intuitively, $F$ can be viewed as a ratio of the channel transmission coefficients associated with the loads $Z_1$ and $Z_2$; more precisely, $F = T_2/T_1$ where $T_i = 2Z_i/(Z_A+Z_i), i = 1, 2$. 
Redefining the estimation problem in terms of $H$ and $F$ renders the observation bilinear in the unknowns. In the channel estimation literature, since $G$ usually results from a random superposition of multipaths, $H$ is usually modeled as a Gaussian random variable. On the other hand, $F$  appears more appropriately modeled as deterministic. This  models the actual estimation problem better than the classical estimation model used in previous studies \cite{Wu,Witt}. In particular, the multi-packet estimators in the prequel assumed constant channel over all packets \cite[Sec. V-B]{Wu}. We consider in the next section the general and more important estimation problem, where the channel is time-varying over multiple packets and follows a complex Gaussian prior distribution. 

\ifcomments
% In practical applications, observing $F$ may actually be more convenient than observing $Z_A$ directly. %In adaptive matching based on gradient descent of $Z_1$, we generally want to test whether $Z_2$ is a better match to $Z_A$ than $Z_1$. Note this question is equivalent to testing whether $|F| >1$, so observing $F$ is actually more convenient than observing $Z_A$ directly.
With these definitions, we can now precisely state the estimation problem considered in this paper. We would like to estimate the
complex-valued parameters
\beq\label{theta}
\boldsymbol{\theta}\ \define \  \begin{bmatrix}
	H & F
\end{bmatrix}^T \ ,
\eeq
based on the observations \er{vL}. Here, $H$ is a zero-mean Gaussian random variable, $H\sim\mccn(0, \sigma_H^2)$ and $F$ is deterministic. Note the prior probability density function (pdf) of $H$,
\beq\label{pdf_h}
p(H) \ = \ \frac{1}{\pi \sigma_H^2} \exp\left(-\frac{|H|^2}{\sigma_H^2}\right) \ ,
\eeq
does not depend on $F$. For simplicity, we  assume $\sigma_H^2$ is known to the receiver. %Although this does not appear consistent with the assumption that $Z_A$ is unknown, it can be justified by observing $\sigma_H^2$ can be accurately estimated from the data received using the receiver impedance $Z_1$. To illustrate, we will relax this assumption later in the numerical results.

In the next section,,  we investigate optimal estimators when multiple packets are available for joint processing, and prior information for channel exists.
\fi

\section{Joint MAP/ML Estimators}
%We now consider estimators that combine observations from multiple training packets. 
Consider a sequence of $L$ packets, where $H_i$ denotes the channel during the $i$-th packet transmission. We assume the channel information vector is unknown, jointly-distributed, complex Gaussian, such that 
\beq\label{H}
\bH \ = \ [H_1 , \ldots , H_L]^T \sim {\cal CN} ( {\bf 0}, \bC_\bH ) \ , 
\eeq
where the channel covariance $\bC_\bH$ is known at the receiver. We assume $F$ changes more slowly with time, and is regarded as fixed over the $L$ packets. 
If each packet is formatted as \er{vL}, then
the entire $L$ packets of observations can be written in matrix form, 
\beq\label{VL}
{\bV}_L  \ = \ \begin{bmatrix}
	\cdots & \bv_{L,i} & \cdots 
\end{bmatrix} \ , ~~~ 1\leq i\leq L  \ .
\eeq
where each column is the observation of the $i$-th packet \er{vL},
\beq\label{vLi}
\bv_{L,i} \ \define \ \begin{bmatrix}
	\bv_{1,i}\\
	\bv_{2,i}
\end{bmatrix}\ = \ \begin{bmatrix}
	H_i\bx_1\\FH_i\bx_2
\end{bmatrix} + \bn_i \ .
\eeq
We assume the noise $\bn_i \sim\mathcal{CN}({0}, \sigma_n^2\bI_T)$ is temporally i.i.d.. The aim  is to use
these observations to estimate the unknowns
\beq\label{theta}
{\btt} \ \triangleq \ \begin{bmatrix}
	\bH\\
	F
\end{bmatrix} \ .
\eeq

Now we have precisely defined the estimation parameters and the observation model, it is important to identify an appropriate estimation framework. 
In classical estimation, unknown parameters are modeled as deterministic; in Bayesian estimation, they are modeled as random variables. The problem described in \eqref{theta} is a \ti{hybrid} estimation problem \cite[pg.~329]{vt}, because $\boldsymbol{\theta}$ contains both random and deterministic parameters. Several authors have investigated hybrid estimation problems  \cite{mes09,mes97,roc,tabr,ren}. Estimators and Cram\'er-Rao-type bounds for hybrid estimation were formulated in \cite{roc} in the context of passive source localization. In this section, we apply similar tools to address the estimation problem formulated in \er{theta}.

\subsection{Estimators for General ${\bf C}_\bH$}
Before discussing estimators, it is convenient to introduce sufficient statistics that summarize the information in the observations $\bV_L$ in \eqref{VL} relevant to estimation of $\btt$. In classical estimation, a statistic ${\bf V}=f(\bV_L)$ is sufficient
to estimate $\btt$ based on $\bV_L$ if $p( \bV_L | {\bV}; \btt) = p( {\bV}_L | {\bf V})$. In Bayesian settings, a statistic is sufficient if $p( {\bV}_L, {\boldsymbol \theta} | {\bf V}) = p( {\bV}_L | {\bf V})p(\btt | {\bf V})$. It is well known that
classical sufficiency implies Bayesian sufficiency; it also clearly implies sufficiency for the hybrid estimation problem, which is defined in an analogous way.

A set of sufficient statistics to estimate $\btt$ based on observations in \er{VL} is derived in the Appendix. We  therefore consider the observations to be the following vector in $\mbbC^{2L}$, 
\bea\label{V}
\bV  & \define &\begin{bmatrix}
	\bV_1\\
	\bV_2
\end{bmatrix} \ = \ \begin{bmatrix}
\bH + \bN_1\\
F\bH + \bN_2
\end{bmatrix} \ , 
\eea
where $\bN_1 \sim {\cal CN} ( {\bf 0}, \frac{\sigma_n^2}{S_1}\bI_L)$ and $\bN_2 \sim {\cal CN} ( {\bf 0}, \frac{\sigma_n^2}{S_2}\bI_L)$ are independent noise vectors, and we define
\beq\label{Ss}
S_1 \ \triangleq \ \bx_1^H\bx_1  \ , ~~~ S_2 \ \triangleq \ \bx_2^H\bx_2 \ .
\eeq 
Conditioned on $\bH$, the sufficient statistic $\bV$ is a complex Gaussian random vector, and its pdf is
\bea\label{pdf_V}
p(\bV|\bH;F) \ = \ \frac{\exp\left[-\left(\bV-\boldsymbol{\mu}\right)^H\bC_\bv^{-1}\left(\bV-\boldsymbol{\mu}\right)\right]}{\det(\pi\bC_\bv)} \ ,
\eea
where the mean and covariance are, respectively,
\beq
 \boldsymbol{\mu} \ = \ \begin{bmatrix}
\bH\\ F \bH
\end{bmatrix} \ ,~~ \bC_\bv \ = \ \begin{bmatrix}
\frac{\sigma_n^2}{S_1}\bI_L & \bzro_{L\times L} \\
\bzro_{L\times L}  & \frac{\sigma_n^2}{S_2}\bI_L
\end{bmatrix} \ .
\eeq
Note $\bzro_{L\times L}$ represents the $L\times L$ all zero matrix. 
Also, we assume $F$ is an unknown constant and $\bH$ is a random vector with pdf
\beq\label{pdf_H}
p(\bH ) \ = \ {[\det(\pi \bC_\bH)]^{-1}}\exp\left(-\bH^H\bC_\bH^{-1}\bH \right)\ . 
\eeq
To jointly estimate $\bH$ and $F$,  we consider estimators that
maximize the hybrid log-likelihood function \cite[pg.~329]{vt}\cite{roc},
\bea
	\hat{\btt}(\bV) \ \triangleq \  \arg \max_{\btt \in \mathbb{C}^{L+1}} \, \ln \, p(\bV,\bH;F) \ , \label{jointMAPMLmp}
	\eea
which are presented in the following theorem.

%\begin{figure*}[t!]
\begin{theorem}[Joint MAP/ML Estimators]\label{th_estimators}
	Given the  sufficient statistic $\bV$  in \eqref{V},
	the joint MAP/ML  estimators for $\bH$ and $F$ are, respectively,
	\beq\label{map}
	\hat{\bf H}_{MAP} \ \triangleq \ \bA\left(\hat{F}_{ML}\right) \left(\bV_1 +\alpha\hat{F}_{ML}^*\bV_2 \right) 	\ ,
	\eeq
	and $\hat{F}_{ML}$, where $\hat{F}_{ML}$ is a zero of the rational function
	\bea\label{ml}
	g(F) &\triangleq& \left( \bV_1 +\alpha F^*\bV_2 \right)^H \bA^H(F)\bA(F)\nn\\
	&&\cdot\left( \bV_2-F\bV_1 +\frac{\sigma_n^2}{S_1} {\bC}_{\bH}^{-1} \bV_2\right) \ , 
	\eea
	and we define
	\beq
	\alpha\ \define \ \frac{S_2}{S_1} \ , ~~~\bA(F) \ \triangleq \ 	\left[ \left( 1+\alpha |F|^2 \right) \bI + \frac{\sigma_n^2}{S_1}\bC_\bH^{-1} \right]^{-1} \ .
	\eeq \hfill 
\end{theorem}
%	\hrulefill
%\end{figure*}%

\begin{proof} From \er{pdf_V} and \er{pdf_H}, we  write the (hybrid) log-likelihood function as
	\bea\label{likelihood}
	&&\mathcal{L}(\btt) \  \triangleq \  \ln \, p(\bV,\bH;F)   =   \ln \, [p(\bV|\bH;F) \, p(\bH) ]\\
	&=&-\frac{S_1}{\sigma_n^2}|\bV_1-\bH|^2-\frac{S_2}{\sigma_n^2}|\bV_2-F \bH|^2 -\bH^H\bC_\bH^{-1}\bH + C  \ ,\nn
	\eea
	where $C$ is a constant, and $S_1$ and $S_2$ are defined in \er{Ss}. If we denote the real and imaginary parts of the parameter vector by $\btt = \btt_r + j \btt_i$, the (complex) gradient is given by \cite[Sec.~15.6]{kay}
	\beq\label{gradient}
	\frac{\partial \mathcal{L}(\btt)}{\partial \btt^*} \ \define \ \frac{1}{2} \left[\frac{\mathcal{L}(\btt)}{\partial \btt_r} + j\frac{\mathcal{L}(\btt)}{\partial \btt_i}\right] \ .
	\eeq
	In order for $\btt$ in \eqref{theta} to maximize the log-likelihood $\mathcal{L}(\btt)$, it is necessary that the gradient $\frac{\partial \mathcal{L}(\btt)}{\partial \btt^*}$ vanishes, i.e.,
	\beq\label{projection}
		%\frac{\partial\mathcal{L}(\btt)}{\partial \btt^*} =
		\begin{bmatrix}
			\frac{S_1}{\sigma_n^2}(\bV_1-\bH)+\frac{S_2}{\sigma_n^2}(\bV_2-F \bH)F^*- {\bf C}_{\bH}^{-1} \bH\\
			\frac{S_2}{\sigma_n^2}\bH^H(\bV_2-F \bH)
		\end{bmatrix}
		\ = \ {\bf 0}   \ .  
	\eeq
Solving the first equation for $\bH$, we obtain \er{map}. 
To find $F$,  substitute \er{map} into the second equation, which yields \er{ml}. $\hfill\diamond$
\end{proof}

Since $g(F)$ is a ratio of polynomials, it may have multiple zeros, so  Theorem \ref{th_estimators} does not necessarily specify unique estimators. Given multiple zeros $F_1, ..., F_m$, however, we can identify the MAP/ML solution: For each $F_j$, we
can, in principle, calculate a corresponding estimate of $\bH_j$ from \er{map}. The MAP/ML solution will be the pair $\hat{\btt_j}=[\bH_j^T,F_j]^T$  that maximizes $\mathcal{L}(\btt)$. 

\ifcomments
We conjecture the ML zero is always the one that maximizes ${\rm Re}[ \bV_2^H \bV_1 F_i ]$, and will test this conjecture in the numerical results.

In principle, Theorem 2 requires us to extract and compare all of the zeros of $g(F)$. However, the form of equations \er{mpmap} and \er{projection} suggests a simpler, alternative approach to approximate the MAP/ML estimators. Note
the second equation in \er{projection} can be rewritten as
\beq
F= {\bf H}^H \bV_2/ {\bf H}^H\bH \ . \label{projection2}
\eeq
Thus, we could attempt to solve the equations iteratively by setting $F_0=0$
and solving for $\bH_0$ using \er{mpmap}. This is equivalent to estimating $\bH$ using only $\bV_1$. Using $\bH_0$, we can then form
improved estimate of $F$, say $F_1$, from \er{projection2}. We can then iterate between these two equations to improve the estimates.
This approach will be also explored in the numerical results.
\fi

\subsection{Special Cases of $\sigma_n^{-2}\bC_\bH$}
We now consider three extreme cases of $\sigma_n^{-2} \bC_\bH$ in which Theorem \ref{th_estimators} yields explicit, closed-form 
estimators for $\bH$ and $F$. First suppose $\bH$ \eqref{H} is an i.i.d. sequence, so  $\bC_\bH= \sigma_H^2 \bI$. Thus, $g(F)=0$ reduces to, 
\bea\label{iidCH}
%\left( \frac{S_1}{\sigma_n^2}\bV_1 +\frac{S_2}{\sigma_n^2}F^*\bV_2 \right)^H \left[ \left( \frac{S_1}{\sigma_n^2}+\frac{|F|^2S_2}{\sigma_n^2} \right) \bI + {\bf R}_{\bH}^{-1} \right]^{-1} \left( \frac{S_1}{\sigma_n^2}(\bV_2-F\bV_1) +{\bf R}_{\bH}^{-1} \bV_2 \right) \ = \ \\
\left( \bV_1 +\alpha F^*\bV_2 \right)^H  \left( \bV_2- cF\bV_1 \right)	\ = \ 0 \ ,
\eea
where 
\beq\label{c}
c \ \define \  {S_1\sigma_H^2}/({S_1\sigma_H^2+\sigma_n^2})\ . 
\eeq
Expanding the product and defining $P_{ij} \triangleq (1/L)\bV_i^H \bV_j$, we obtain
\bea\label{Fquadratic}
P_{12} + (\alpha P_{22}- c P_{11} )F - \alpha c  P_{21} F^2 \ = \ 0 \ .
	\eea
%If we multiply both sides by $-P_{21} = -P_{12}^*$ and define $E \triangleq P_{21}F$, this equation becomes 
%\bea\label{Equadratic}
%\alpha c  E^2- (\alpha P_{22}- c P_{11} )E-|P_{12}|^2  \ = \ 0 \ .
%	\eea
%This is a quadratic equation with real coefficients, which has two real roots:
%\beq\label{Eroots}
%E = \frac{\alpha P_{22}- c P_{11} \pm \sqrt{ (\alpha P_{22}-P_{11} )^2 + 4\alpha c |P_{12}|^2 }}{2 \alpha c} \ .
%\eeq
The ML estimate of $F$ is given by one of two roots,
\beq\label{Froots}
F = \frac{\alpha P_{22}-cP_{11} \pm \sqrt{ (\alpha P_{22}-cP_{11} )^2 + 4\alpha c |P_{21}|^2 }}{2c\alpha P_{21}} \ .
\eeq
We conjecture the positive root above is always the joint ML estimate; however, as noted earlier, we can determine which root is the MAP/ML solution by identifying the root that maximizes $\mathcal{L}(\btt)$.

The second extreme case we consider is an arbitrary non-singular ${\bf C}_{\bH}$ in the low noise limit, $\sigma_n^2 \rightarrow 0$.
In \er{ml}, applying the approximation $\sigma_n^2 {\bf C}_{\bH}^{-1} \approx {\bf 0}$,  $g(F)=0$ leads to
\bea\label{highsnriid}
\left( \bV_1 +\alpha F^*\bV_2 \right)^H  \left( \bV_2-F\bV_1 \right)	\ = \ 0 \ .
	\eea
Comparing this to \er{iidCH}, we see this equation is identical to the case of i.i.d. $\bH$ for $c=1$. It follows immediately that the ML estimate of $F$ is given by one of the following two roots,
\beq\label{Froots2}
F = \frac{\alpha P_{22}-P_{11} \pm \sqrt{ (\alpha P_{22}-P_{11} )^2 + 4\alpha  |P_{21}|^2 }}{2\alpha P_{21}} \ .
\eeq
This result suggests that the ML estimator for arbitrary non-singular ${\bf C}_{\bH}$ is asymptotically the same as the i.i.d. ML estimator in \er{Froots} in the low-noise or high SNR limit.

Finally, the last extreme case we consider
is the single-packet case, i.e., $L=1$, where $\bH=H\sim\mccn(0,\sigma_H^2)$ is a special case of \eqref{pdf_H}.  Under these assumptions, the sufficient statistic in \er{V} becomes $\bV=[V_1, V_2]^T$, and Theorem \ref{th_estimators} leads to the single-packet joint MAP/ML solution,
	\beq\label{estimators_sp}
	\hat{\btt}_{SP} \ = \ \begin{bmatrix}
		\hat{H}_{MAP}& \hat{F}_{ML}
	\end{bmatrix}^T \ = \ \begin{bmatrix}
	c V_1&\frac{V_2}{cV_1}
\end{bmatrix}^T \ , 
	\eeq
where $c$ is defined in \eqref{c}.  
Note two sets of solutions exist for \er{projection}, i.e.,  $\hat{\btt}_{SP}$ as in \eqref{estimators_sp} and $\hat{\btt_2} = [0, -V_1^*/(\alpha V_2^*)]^T$. Substituting each back into the single-packet log-likelihood, it can be shown that
\beq
\mathcal{L}(\hat{\btt}_{SP})  =   -\frac{S_1}{\sigma_n^2}|V_1|^2 ( 1-c ) \geq 
\mathcal{L}(\hat{\btt_2})  =   -\frac{S_1}{\sigma_n^2}|V_1|^2-\frac{S_2}{\sigma_n^2} |V_2|^2  \ , \nn
\eeq
for all $V_1, V_2, 0 < c < 1$. So \eqref{estimators_sp} is the global maximum of $\mathcal{L}(\btt)$, and thus the single-packet joint MAP/ML estimators. 

% slow fading case
The extremely slow fading case is mathematically identical to the single-packet case, where $\bH=H{\bf 1}$, and $\bf 1$ is the all one vector. Its joint MAP/ML solution follows directly from \er{estimators_sp}, by regarding all $L$ packets as one large packet of size $LT$.

\section{Properties of the Joint MAP/ML Estimators}
The performance of estimators is often measured by low-order central moments, such as bias and mean-squared error (MSE). In this section, we explore the behavior of these moments for the joint MAP/ML estimators in Theorem \ref{th_estimators}, as well as the consistency of these estimators.

\subsection{Single-Packet Estimators}
First consider the single-packet estimators in \eqref{estimators_sp}. Clearly $\hat{H}_{MAP}$ is complex Gaussian, and has finite MSE. 
But $\hat{F}_{ML}$ in \er{estimators_sp} is a ratio of two joint complex Gaussian random variables with zero mean.  The mean of general complex Gaussian ratios has been derived in closed-form\cite[eq.~3]{cgr}. Based on this result, we show the joint ML estimator in $\hat{F}_{ML}$ is unbiased, i.e., $E[\hat{F}_{ML}] = F$ for all $F\in \mbbC$.  However, its MSE,  $E[|\hat{F}_{ML}-F|^2]$, is unbounded. Consequently, mean absolute error (MAE) is used instead of MSE in the simulations when the single-packet $\hat{F}_{ML}$ \er{estimators_sp} is plotted (e.g., see Fig. \ref{fig_F_MAE_corrH}). %Ratios of Gaussian variables often have heavy tails, in the sense that no finite moments exist \cite{ast,hyd,lai,wu}. For example, if $U$ and $V$ are independent, real, standard Gaussian random variables, the ratio $U/V$ is a standard Cauchy random variable, which has no finite moments of order greater than or equal to one; only fractional moments exist.

\subsection{Consistency}
Except for the single-packet (or slow-fading) case \er{estimators_sp}, it appears difficult to evaluate the mean of $\hat{F}_{ML}$ for $L > 1$, although numerical examples suggest it may be biased (see Fig. \ref{fig_Fc_bias}). To gain insight into its bias, we study the consistency of $\hat{F}_{ML}$ for i.i.d. $\bH$ when $L$ becomes large. By definition, an estimator $\hat{F}$ is consistent if, given any $\epsilon > 0$ \cite[pg. 200]{kay}, it converge in probability to the true parameter, 
$\lim\limits_{L\rightarrow \infty} \Pr(|\hat{F}-F| > \epsilon) =  0$. 
As $L \rightarrow \infty$, the coefficients in \er{Fquadratic} converge in probability,
\beas
\lim_{L \rightarrow \infty}P_{11} & = & \sigma_H^2 + \sigma_n^2/S_1 \ , ~~ 
\lim_{L \rightarrow \infty} P_{21} \ =  \ \sigma_H^2 F^* \ ,\nn\\
\lim_{L \rightarrow \infty} P_{22} & = & |F|^2 \sigma_H^2 + \sigma_n^2/S_2 \ ,
\eeas
so the roots \er{Froots} converge to
\bea\label{largeLroots}
\left[\alpha|F|^2  - cd  \pm \sqrt{ (\alpha|F|^2  -cd )^2 + 4\alpha c |F|^2 }\right]/{2\alpha c F^*} \ , 
%& = & \frac{\alpha(|F|^2 \sigma_H^2 + \sigma_n^2/S_2)-c(\sigma_H^2 + \sigma_n^2/S_1) \pm \sqrt{ (\alpha P_{22}-cP_{11} )^2 + 4\alpha c |P_{12}|^2 }}{2c\alpha \sigma_H^2 F^*} \\
\eea
where $d \triangleq  1 - (\sigma_n^2/S_1 \sigma_H^2)^2$. %so $cd \sigma_H^2 = c(\sigma_H^2 + \sigma_n^2/S_1)- \alpha \sigma_n^2/S_2$. 
Since neither root converges to $F$, it follows the joint ML estimator  $\hat{F}_{ML}$
is inconsistent\footnote{The joint ML estimator in \er{Froots} should not be confused with  the true ML in classical estimation, where the latter is provably consistent yet the former is not; see Kay and the references therein \cite[pg. 211]{kay}.} in $L$. However, if the second term in the square root in \eqref{largeLroots} is multiplied by $d$ (to complete the square), the positive root becomes  $F/c$. 
This suggests a consistent estimator
\beq\label{FC}
\hat{F}_{C} \ \triangleq \ \frac{\alpha P_{22}-cP_{11} + \sqrt{ (\alpha P_{22}-cP_{11} )^2 + 4\alpha c d  |P_{21}|^2 }}{2\alpha P_{21}}  \ .
\eeq
For large signal-to-noise ratios ($S_1 \sigma_H^2 \gg \sigma_n^2$), we note $c$ and $d$ are approximately 1, and the positive root in \er{Froots} coincides with this consistent estimator. In the next section, we will compare the performance of $\hat{F}_{ML}$ and $\hat{F}_{C}$ via simulations.

\subsection{Hybrid Cram\'{e}r-Rao Bound}
To evaluate estimator performance, it is useful to consider a fundamental lower bound on
the error covariance
\beq
{\bf C}_{\hat{\boldsymbol \theta}} \ \triangleq \ E_{\bV,\bH;F} \left[ \left(\hat{\boldsymbol \theta}-{\boldsymbol \theta} \right)\left(\hat{\boldsymbol \theta}-{\boldsymbol \theta} \right)^H \right] \ ,
\eeq
where $E_{\bV,\bH;F}$ denotes expectation with respect to the pdf $p( \bV ,\bH;F)$, which is the product of \er{pdf_V} and \er{pdf_H}. A Cram\'{e}r-Rao Bound for
the hybrid estimation problem was presented in \cite[pg.~329]{vt},\cite{roc}. Let $\mathcal{L}(\btt)$ be the hybrid log-likelihood in \er{likelihood}, and define the pseudo-information as
\beq\label{pseudo_info}
E_{\bV,\bH;F}\left[\left(\frac{\partial\mathcal{L}(\btt)}{\partial \btt^*}\right)\left(\frac{\partial\mathcal{L}(\btt)}{\partial \btt^*}\right)^T\right] \ .
\eeq
The hybrid CRB (HCRB) states that, if the pseudo-information vanishes, then\footnote{The HCRB is given for real parameters in \cite[pg.~329]{vt}. Here we use the approach described in \cite[Sec.~15.7]{kay} to state this bound in an equivalent form convenient for complex parameters.}
\beq
{\bf C}_{\hat{\boldsymbol \theta}} %\ \triangleq \ E_{\bV_1,\bV_2,\bH;F} \left[ \left(\hat{\boldsymbol \theta}-{\boldsymbol \theta} \right)\left(\hat{\boldsymbol \theta}-{\boldsymbol \theta} \right)^H \right] 
\ \geq \ \boldsymbol{\mathcal{I}}^{-1}  \ ,
\eeq
for any $\hat{\boldsymbol \theta} = [ \hat{\bH}^T, \hat{F} ]^T$ such that $\hat{F}$ is unbiased, where 
$\boldsymbol{\mathcal{I}}$ is the hybrid information matrix,
\beq
\boldsymbol{\mathcal{I}} \ \define \ E_{\bV,\bH;F}\left[\left(\frac{\partial\mathcal{L}(\btt)}{\partial \btt^*}\right)\left(\frac{\partial\mathcal{L}(\btt)}{\partial \btt^*}\right)^H\right]\ . 
\eeq

For the estimation problem \er{theta}, from 
\er{projection} it is easy to verify that the pseudo-information \er{pseudo_info} vanishes, and the HCRB is
\beq\label{hcrb}
{\bf C}_{\hat{\boldsymbol \theta}} \ \geq \begin{bmatrix}
	\left[ \left( \frac{S_1+|F|^2S_2}{\sigma_n^2}\right) \bI + \bC_\bH^{-1}\right]^{-1} & {\bf 0}_{L \times 1}\\
	{\bf 0}_{1 \times L}& \frac{\sigma_n^2}{S_2{\rm Tr}[\bC_\bH]}
\end{bmatrix} \ .
\eeq

Note that uncertainty in $\bH$ affects the estimators differently. For example, doubling all eigenvalues of $\bC_\bH$ makes $p(\bH)$ less informative, and increases the lower bound on $\hat{\bH}$ error covariance, but decreases the error bound on $F$ estimation. This is intuitively reasonable, since a larger $\Trace[\bC_\bH]$ essentially increases the signal-to-noise of the observation of $F$.

\section{Numerical Results}
\begin{figure}[t!]
	\begin{center}
		\includegraphics[width=.5\textwidth, keepaspectratio=true]{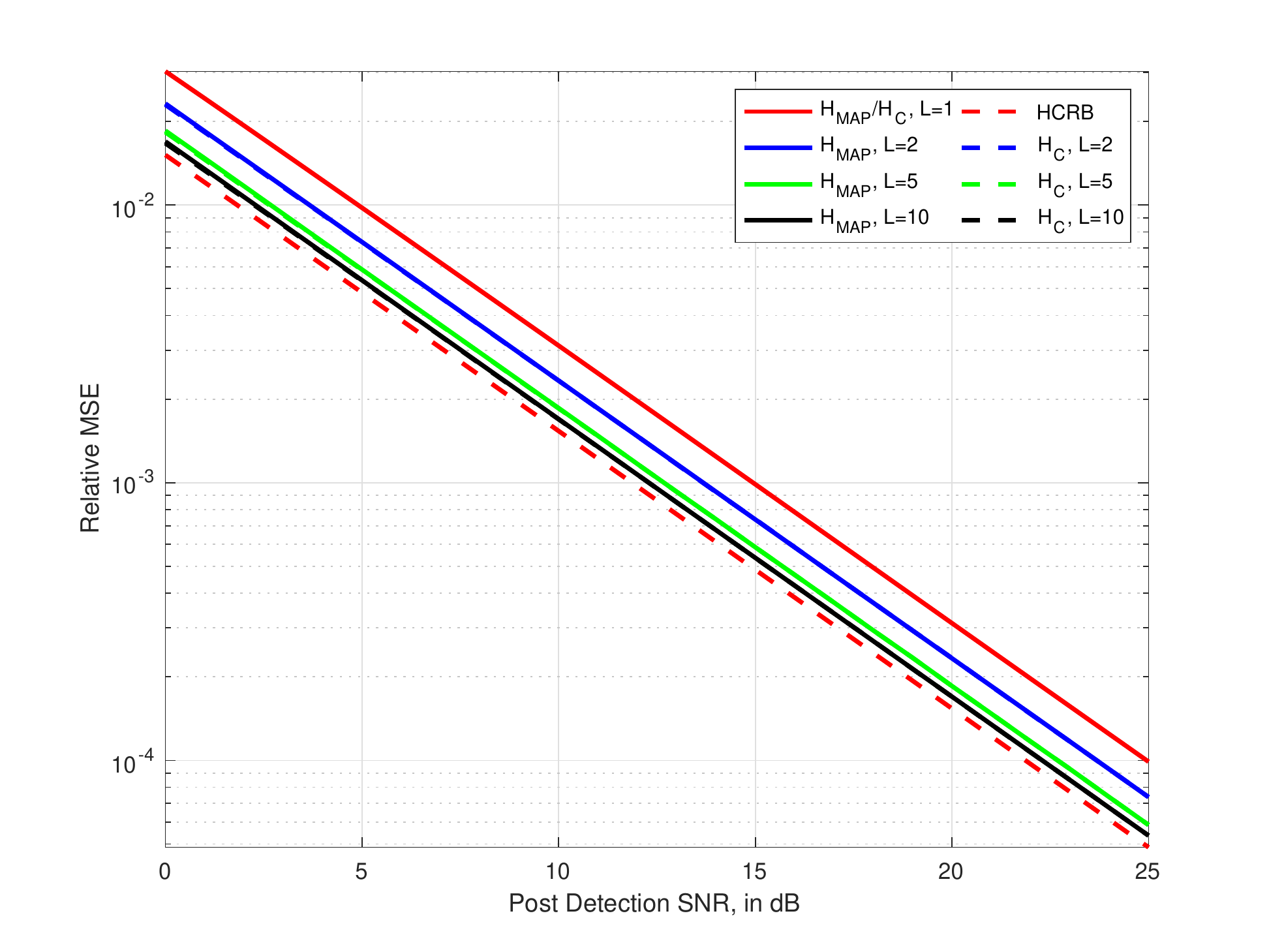}
	\end{center}
	\vspace{-16pt}
	\caption{Relative MSE of $\hat{\bH}_{MAP}$ versus SNR.}
	\vspace{-12pt}
	\label{fig_iidH}
\end{figure}

\begin{figure}[t!]
	\begin{center}
		\includegraphics[width=.5\textwidth, keepaspectratio=true]{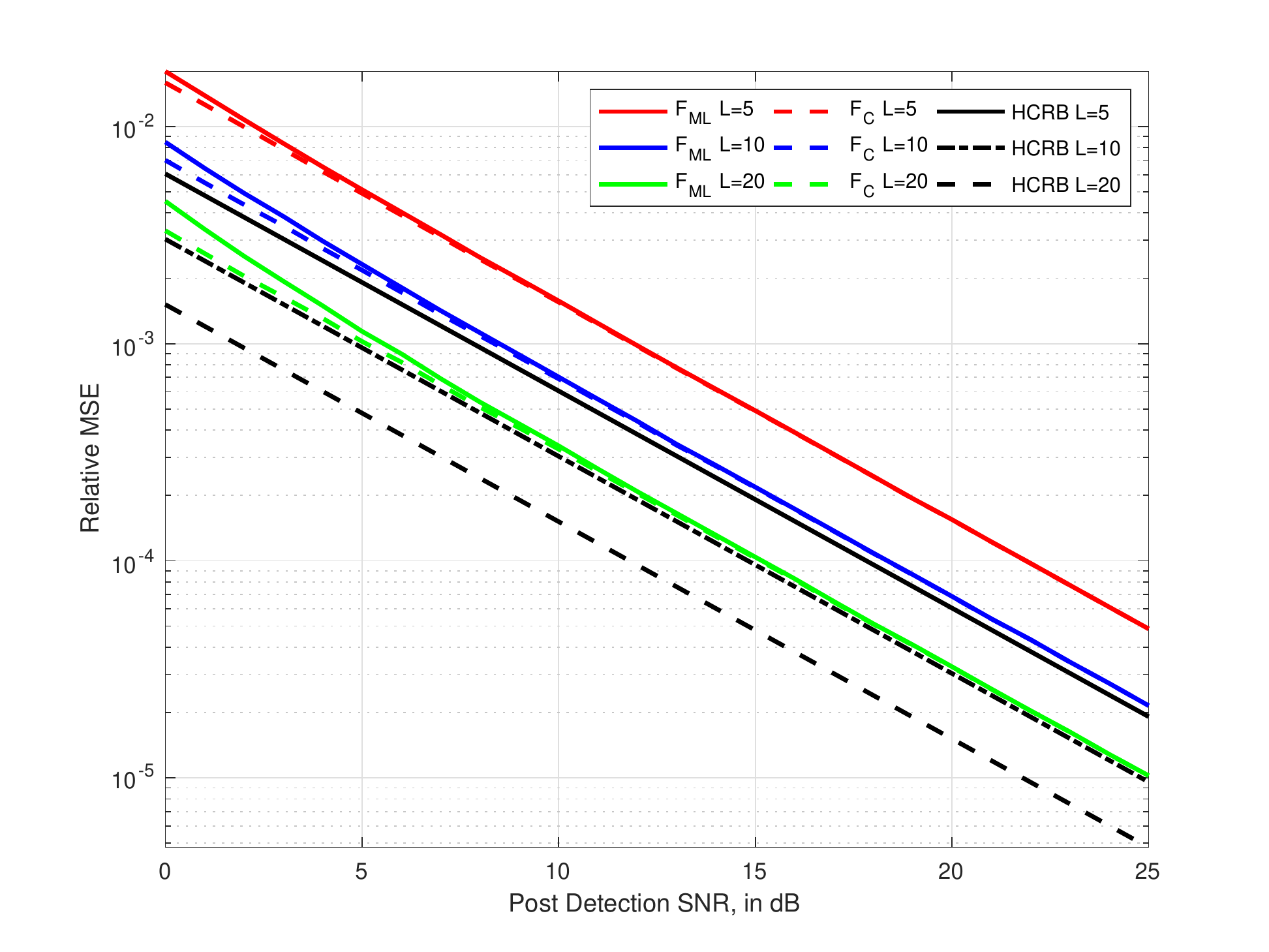}
	\end{center}
	\vspace{-16pt}
	\caption{MSE of $\hat{F}_{ML}$ and $\hat{F}_C$ against SNR for $L>2$.}
	\vspace{-12pt}
	\label{fig_F_CRB}
\end{figure}
In this section, we explore the performance of the estimators in Secs.~III-IV through numerical examples. We take the training sequence in \er{vL} as a unit-magnitude Zadoff-Chu sequence of length $T=64$. The unknown antenna impedance is that of a dipole $Z_A=73+j42.5 \,\Omega$. The load impedance \er{impedanceshifts} is $Z_1=50 \,\Omega$ for the first $K=T/2=32$ symbols of each training sequence, and $Z_2=50+j20 \,\Omega$ for the remaining symbols. From \er{F}, it follows $F=0.9860 + j0.2445$. Suppose the channel $\bH$ is an i.i.d. sequence, so ${\bf C}_\bH = \sigma_H^2 \bI$, and
define the post-detection signal-to-noise ratio (SNR) of a training symbol as
\beq
\rho \ \triangleq \ {\sigma_H^2}/{\sigma_n^2} \ .
\eeq

In Fig.~\ref{fig_iidH}, we plot the relative mean-squared error (MSE) of the MAP estimator \er{map}, which is defined as $E[ \parallel \hat{\bH}_{MAP}- \bH \parallel^2 ]/L \sigma_H^2$, versus $\rho$ for $L=1,2,5,10$ training packets. Here $\hat{F}_{ML}$ is the positive root of \er{Froots}. Also shown is the  HCRB bound on $\hat{\bH}_{MAP}$, which equals the diagonal elements of the upper left matrix block in HCRB \er{hcrb}. This HCRB also lower bounds any channel estimator given $F$. Note the single-packet estimator \er{estimators_sp} is 3 dB away from the HCRB, since it uses only half of the training symbols. When even a few packets are combined, however, the improved estimate of $F$ enables  $\hat{\bH}_{MAP}$ to quickly approach the HCRB to within a fraction of a dB. In particular, for $L=10$, the efficiency is above 90\% for all SNR plotted.

If we assume the MSE of $\hat{F}_{ML}$ is finite for $L >2$, we can compare its MSE performance to  HCRB derived in \er{hcrb}, which, however, is often loose and unachievable\cite{mes09}. 
In Fig.~\ref{fig_F_CRB},
we plot the relative MSE,  $E[ | \hat{F}_{ML}-F|^2]/|F|^2$ versus SNR for $L=5,10,20$. %For comparison, we also plot the classical CRB on $\hat{F}_{ML}$ derived from the marginal pdf $p(\bV;F)$.
%Note that $\hat{F}_{ML}$ is less than a dB from the classical CRB for high SNRs for $L=5$, and nearly achieves the HCRB for $L=10$ and 20. 
At low SNRs, we observe that $\hat{F}_{ML}$ diverges somewhat from the straight line it traces at high SNR. One possible explanation for this behavior is the presence of bias in $\hat{F}_{ML}$ at low SNRs. Some support for this hypothesis was given in Sec.~IV-B, where we showed that $\hat{F}_{ML}$ is inconsistent in $L$. To remedy this situation, we defined a consistent estimator
$\hat{F}_{C}$ in \er{FC},  which is also plotted in Fig.~\ref{fig_F_CRB}. We observe that
$\hat{F}_C$ performs as well or better than $\hat{F}_{ML}$ for all SNR and $L$ in the figure. In particular, $\hat{F}_C$ outperforms $\hat{F}_{ML}$ at low SNR. Back in Fig. \ref{fig_iidH}, we also plotted $\hat{\bH}_C$, which substitutes $\hat{F}_C$ in \er{map} instead of $\hat{F}_{ML}$. However, both channel estimators appear to coincide for all SNR and $L$ plotted.

To better understand the impact of bias, in Fig.~\ref{fig_Fc_bias} we plot the absolute relative bias, defined as
$| E [\hat{F}_{ML} -F ]/{F}|$, versus SNR for $L=5,10$ and two values of $F$, i.e., $F_1 =  0.9860 + j0.2445$ and $F_2 = 1.0644 + j0.5451$. 
Note all four curves for $\hat{F}_{ML}$ seem to coincide, and suggest a relative bias that decreases with SNR but does not appear to depend on $L$ or $F$. For comparison, we also plot the corresponding curves for the consistent estimator, $\hat{F}_C$.
Compared with $\hat{F}_{ML}$, the bias of $\hat{F}_C$ at low and medium SNR appears two orders of magnitude smaller, and does not suggest dependence on $L$. We do not know if $\hat{F}_C$ is unbiased, but it appears less biased than $\hat{F}_{ML}$. 
\begin{figure}[t!]
	\begin{center}
		\includegraphics[width=.5\textwidth, keepaspectratio=true]{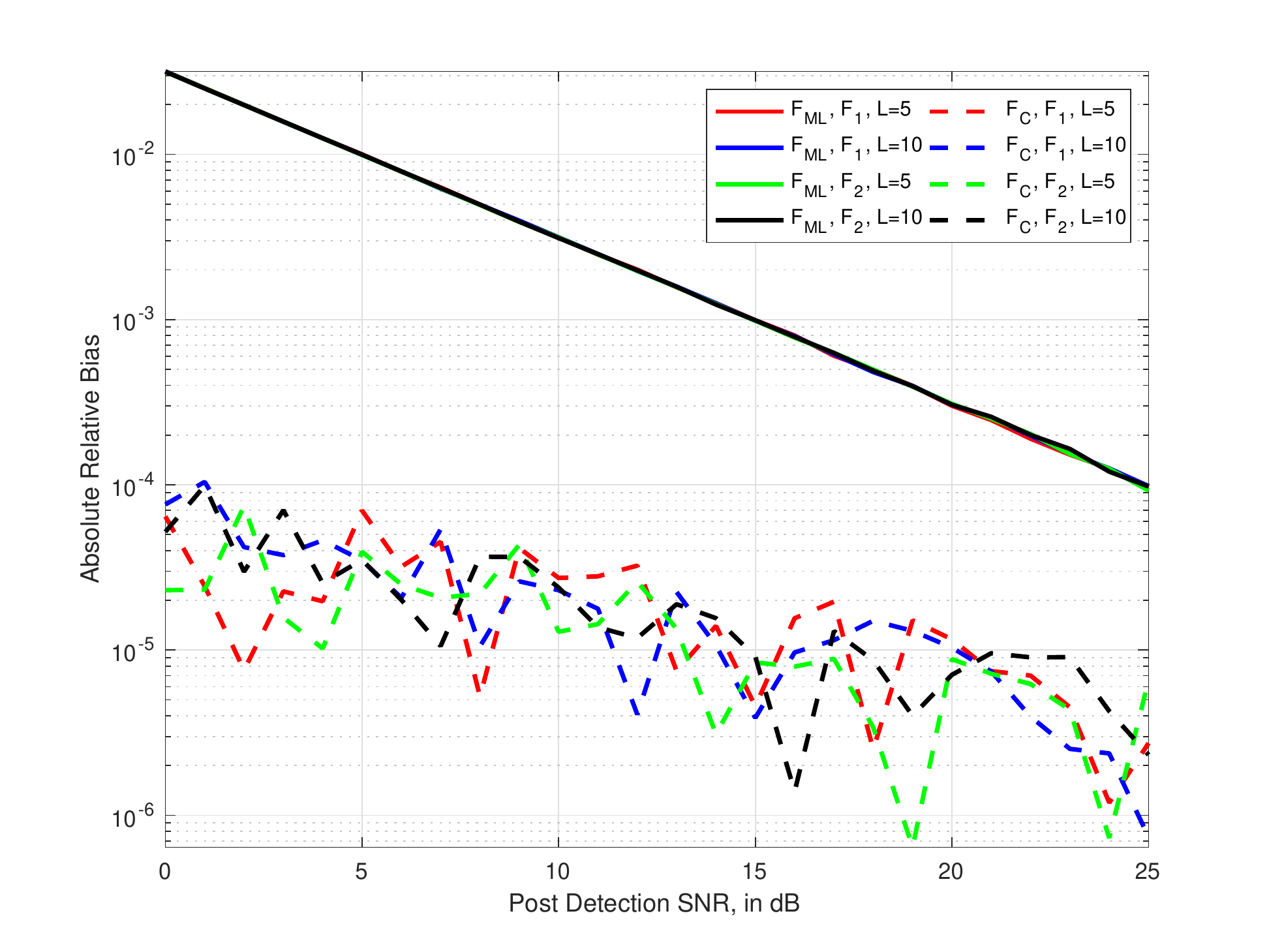}
	\end{center}
	\vspace{-16pt}
	\caption{Relative Bias of $\hat{F}_{ML}$ versus SNR and $L$. }
	\vspace{-12pt}
	\label{fig_Fc_bias}
\end{figure}

Thus far, we have assumed $\bH$ is i.i.d. Now consider an example of extremely slow fading, where $H_1 = \cdots = H_L$. 
In Fig.~\ref{fig_H_corrH}, we plot the relative MSE of $\hat{H}_{MAP}$ versus SNR for several values of $L$. The performance of $\hat{\bH}_{MAP}$ \er{map} for i.i.d. $\bH$ is also included for comparison. Not surprisingly, strong correlation in $\bH$ leads to a 
smaller MSE for $L> 1$, since each channel is now averaged over $L$ observations. However, it is worth noting the {\em opposite} is true of $\hat{F}_{ML}$. In Fig.~\ref{fig_F_MAE_corrH}, we plot the relative MAE of $\hat{F}_{ML}$, i.e., $E[|\hat{F}_{ML}-F|]/|F|$,  under the same conditions as Fig.~\ref{fig_H_corrH}. Here the i.i.d. channel leads to considerably better estimates of $F$, since multiple, independent channel observations average out variations due to $\bH$.
\begin{figure}[t!]
	\begin{center}
		\includegraphics[width=.5\textwidth, keepaspectratio=true]{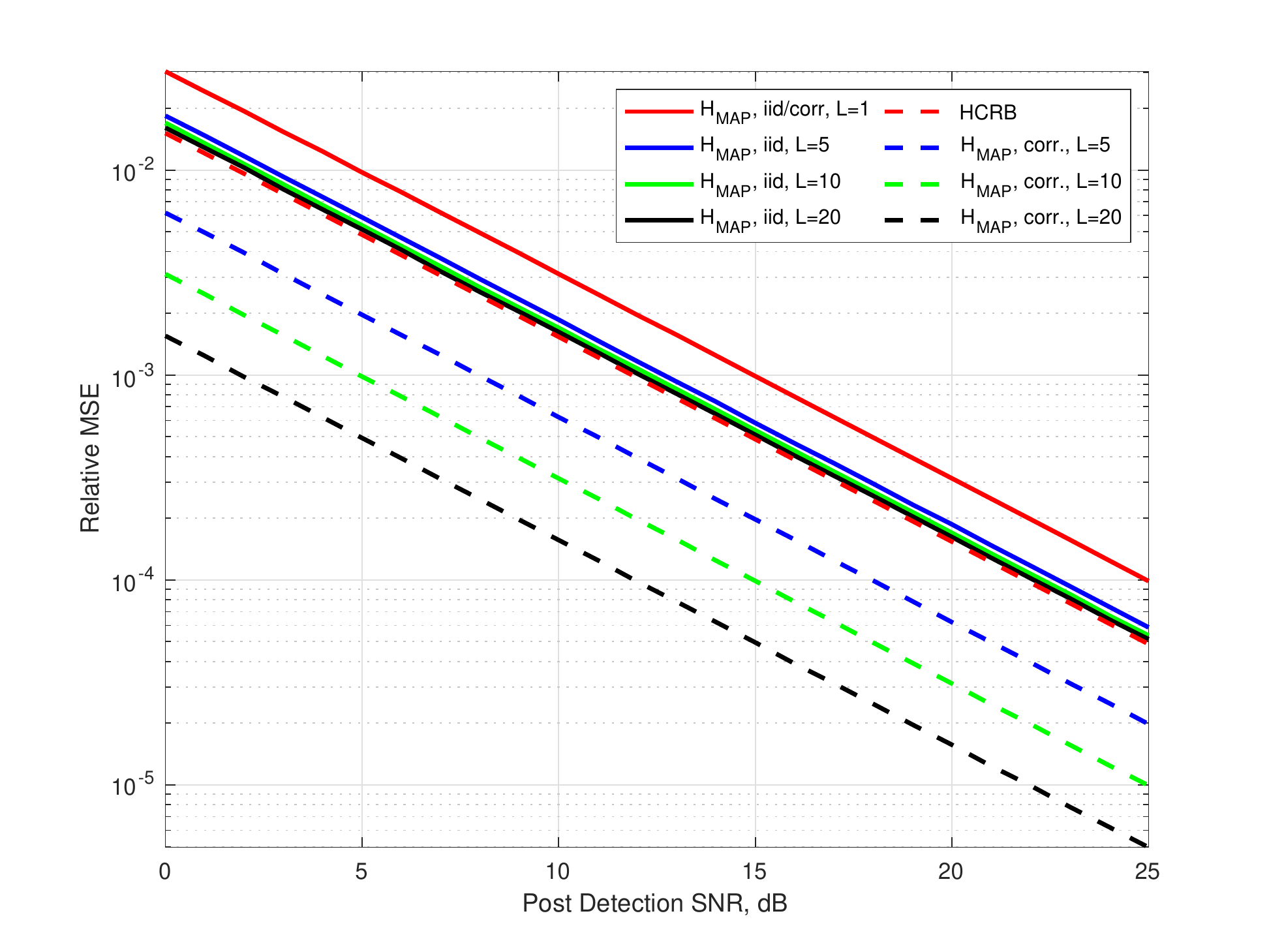}
	\end{center}
	\vspace{-16pt}
	\caption{Relative MSE of $\hat{H}_{MAP}$ versus SNR for a correlated channel.}
	\vspace{-12pt}
	\label{fig_H_corrH}
\end{figure}
\begin{figure}[t!]
	\begin{center}
		\includegraphics[width=.5\textwidth, keepaspectratio=true]{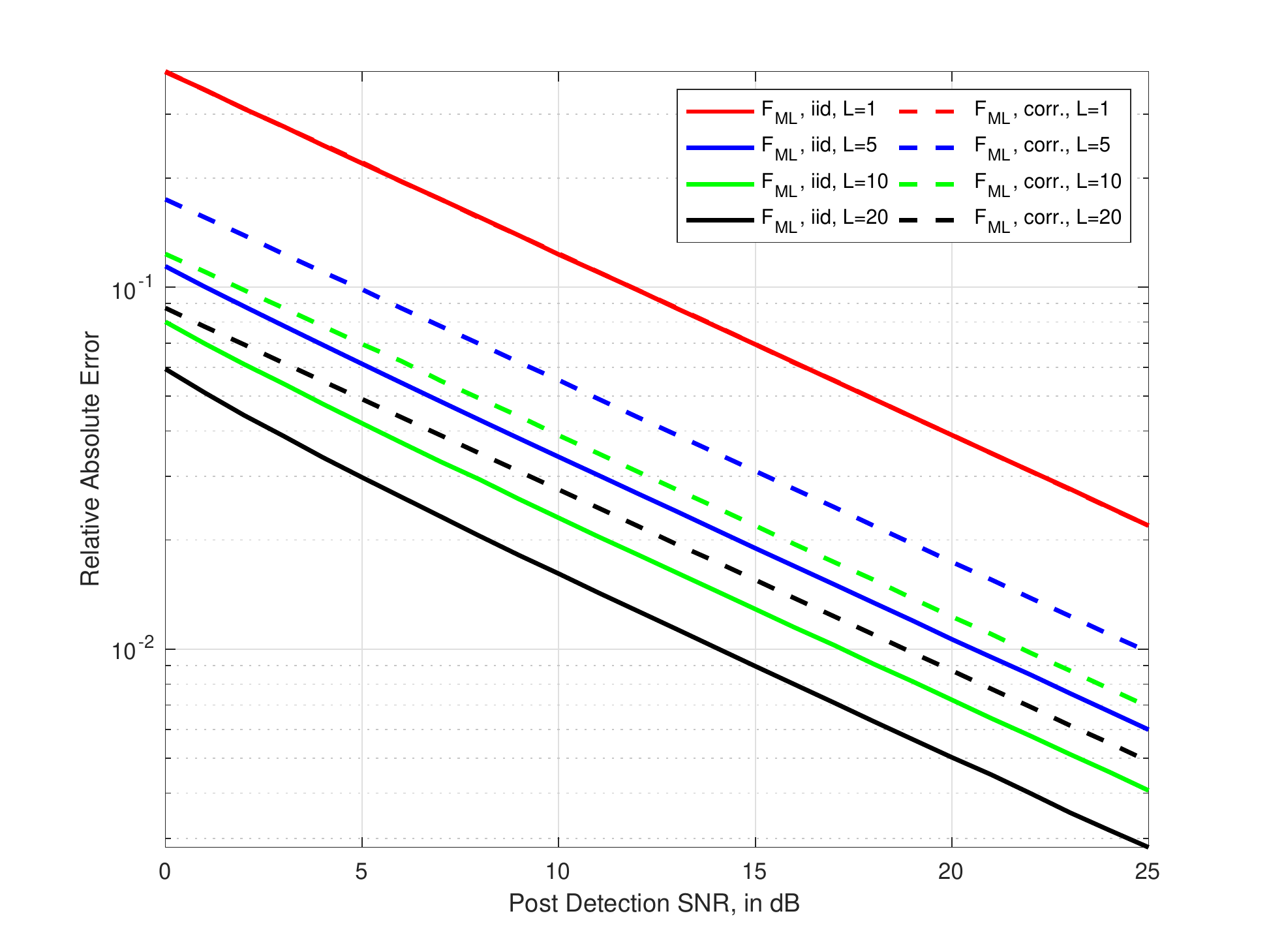}
	\end{center}
	\vspace{-16pt}
	\caption{Relative MAE of $\hat{F}_{ML}$ versus SNR for a correlated channel.}
	\vspace{-12pt}
	\label{fig_F_MAE_corrH}
\end{figure}

\section{Conclusions}
In this paper, we developed a hybrid estimation framework for joint estimation of channel information and antenna impedance. Joint maximum a posteriori and maximum-likelihood (MAP/ML) estimators are derived for multi-packet scenarios in a temporally correlated fading channel. 
Important properties, e.g., efficiency and/or bias, are investigated for the joint MAP/ML estimators, either analytically or through numerical means. We found the joint ML estimators for $F$ are generally biased, except for the single-packet case. After studying its consistency, we found a consistent estimator, which performs as well or better than the joint ML estimator in terms of bias and MSE. Furthermore, the joint MAP channel estimator becomes efficient with a sufficient number of packets. We also explored the impact of channel correlation on estimation accuracy to find that correlation improves channel estimation but worsens impedance estimation. %All in all, numerical results suggest that antenna impedance can be accurately estimated, in exchange for a small, controlled increase in channel estimation error.

%\newpage
\appendix
	Rearrange the conditional PDF $p\left(\bv_{L,i};\btt\right) $ of any arbitrary $\bv_{L,i}$ defined in \eqref{vLi},   $1\leq i \leq L$, 
	\bea
	 \exp \left(-\frac{S_1|V_{1,i}-H|^2}{\sigma_n^2}-\frac{S_2|V_{2,i}-FH|^2}{\sigma_n^2}\right)g(\bv_{L,i}) \nn\ ,
	\eea
	where $g(\bv_{L,i})$ only depends on $\bv_{L,i}$ (and known training sequences $\bx_1$ and $\bx_2$), but not $\btt$,
	\bea
	g(\bv_{L,i}) \ \define \ {(\pi \sigma_n^2)^{-T} }
	\exp\left[\frac{1}{\sigma_n^2}(V_{1,i}+V_{2,i}-{|\bv_{L,i}|^2})\right] \ ,\nn
	\eea
	and we define
	\beq
	V_{1,i} \ \define \ {\bx_1^H\bv_{1,i}}/{S_1} \ , ~~~ V_{2,i} \ \define \ {\bx_2^H\bv_{2,i}}/{S_2} \ .
	\eeq
	By  Neyman-Fisher factorization Theorem \cite[pg. 117]{kay}, $V_{1,i}$ and $V_{2,i}$ is a set of sufficient statistic. To get \er{V}, combine all $L$ sets of sufficient statistics in vector form, i.e., $\bV_1=[\cdots~V_{1,i}~ \cdots]^T$, and $\bV_2=[\cdots~V_{2,i}~ \cdots]^T$. 
	$\hfill\diamond$

%	References
\bibliographystyle{unsrt}

\end{document}